\documentclass[runningheads]{llncs}
\usepackage{graphicx}
\usepackage{amsmath}
\usepackage{amssymb}
\begin{document}
\title{When Can $l_p$-norm Objective Functions be Minimized via Graph Cuts?}
\titlerunning{$l_p$-norm Minimization by Graph Cuts}
%
\author{Filip Malmberg\and Robin Strand}
\authorrunning{ F. Malmberg \and R. Strand}
%

\institute{Department of Information Technology, Uppsala University, Sweden\\
\email{filip.malmberg@it.uu.se}, \email{robin.strand@it.uu.se}}

\maketitle              
\begin{abstract}
Techniques based on minimal graph cuts have become  a standard tool for solving combinatorial optimization problems arising in image processing and computer vision applications. These techniques can be used to minimize objective functions written as the sum of a set of unary and pairwise terms, provided that the objective function is \emph{submodular}. This can be interpreted as minimizing the $l_1$-norm of the vector containing all pairwise and unary terms. By raising each term to a power $p$, the same technique can also be used to minimize the $l_p$-norm of the vector. Unfortunately, the submodularity of an $l_1$-norm objective function does not guarantee the submodularity of the corresponding $l_p$-norm objective function. The contribution of this paper is to provide useful conditions under which an $l_p$-norm objective function is submodular for all $p\geq 1$, thereby identifying a large class of $l_p$-norm objective functions that can be minimized via minimal graph cuts.

\end{abstract}
\section{Introduction}

Many fundamental problems in image processing and computer vision, such as image filtering, segmentation, registration, and stereo vision,  can naturally be formulated as optimization problems~\cite{kolmogorov2004energy}. Often, these optimization problems can be described as \emph{labeling} problems, in which we wish to assign to each image element (pixel) an element from some finite set of labels. The interpretation of these labels depend on the optimization problem at hand. In image segmentation, the labels might indicate object categories. In registration and stereo disparity problems the labels represent correspondences between images, and  in image reconstruction and filtering the labels represent intensities in the filtered image. We seek a label assignment configuration $\mathbf{x}$ that minimizes a given objective function $E$, which in the ``canonical'' case can be written as follows: 

\begin{equation}
\label{eq:canonical}
E( \mathbf{x} )  = \sum_{i\in \mathcal{V}} \phi_i(x_i)+\sum_{i,j\in \mathcal{E}} \phi_{ij}(x_i,x_j)  \; .
\end{equation}

\noindent In Eq.~\ref{eq:canonical} above, $\mathcal{G}=(\mathcal{V}, \mathcal{E})$ is an undirected graph and $x_i$ denotes the label of vertex $i\in \mathcal{V}$ which must belong to a finite set of integers $\{0,1\ldots, K-1\}$. We assume that both the \emph{unary} terms $\phi_i(\cdot)$ and the \emph{pairwise} terms $\phi_{ij}(\cdot,\cdot)$ are non-negative for all $i,j$. We seek a labeling $\mathbf{x}=(x_1, \ldots, x_{|\mathcal{V}|})$ for which $E(\mathbf{x})$ is minimal. 

Finding a globally optimal solution to the labeling problem described above is NP-hard in the general case~\cite{kolmogorov2004energy}, but there are classes of objective functions for which efficient algorithms exist. Specifically, for the binary labeling problem, with $K=2$, a globally optimal solution can be computed by solving a max-flow/min-cut problem on a suitably constructed graph, provided that all pairwise terms are \emph{submodular}~\cite{boykov2001fast,kolmogorov2004energy}. A pairwise term $\phi_{ij}$ is said to be submodular if

\begin{equation}
 \phi_{ij}(0,0)+\phi_{ij}(1,1)\leq \phi_{ij}(0,1)+\phi_{ij}(1,0) \; .
\end{equation}

For $K>2$, the optimization problem cannot in general be solved directly via graph cuts. The multi-label problem can, however, be reduced to a sequence of binary valued labeling problems using, e.g., the \emph{expansion move} algorithm proposed by Boykov et al.~\cite{boykov2001fast}. The output of the expansion move algorithm is a labeling that is locally optimal in a strong sense, and that is guaranteed to lie within a multiplicative factor of the global minimum~\cite{boykov2001fast,kolmogorov2004energy}. With this in mind, we here restrict our attention to the binary label case, i.e., $K=2$.

Looking again at the labeling problem described above, we can view the objective function in Eq.~\ref{eq:canonical} as consisting of two parts:

\begin{itemize}
\item A \emph{local} error measure, in our case defined by the unary and pairwise terms.
\item A \emph{global} error measure, aggregating the local errors into a final score. In the case of Eq.~\ref{eq:canonical}, the global error measure is obtained by summing all the local error measures.
\end{itemize} 

The choice of global error measure determines how local errors will be distributed in the optimal solution. Since we assume all terms to be non-negative, minimizing $E$ can be seen as minimizing the $l_1$-norm of the vector containing all unary and pairwise terms. Here, we consider the generalization of this result to arbitrary $l_p$-norms, $p\geq 1$, and thus seek to minimize

\begin{equation}
\label{eq_pnorm_problem}
\left(\sum_{i\in \mathcal{V}} \phi_i^p(x_i)+\sum_{i,j\in \mathcal{E}} \phi^p_{ij}(x_i,x_j) \right)^{1/p} \; ,
\end{equation}

\noindent where $\phi_i^p(\cdot)=(\phi_i(\cdot))^p$ and $\phi^p_{ij}(\cdot,\cdot)=(\phi_{ij}(\cdot,\cdot))^p$. The value $p$ can be seen as a parameter controlling the balance between minimizing the overall cost versus minimizing the magnitude of the individual terms. For $p=1$, the optimal labeling may contain arbitrarily large individual terms as long as the sum of the terms is small. As $p$ increases, a larger penalty is assigned to solutions containing large individual terms. In the limit as $p$ goes to infinity, the global error measure will approach the $L_\infty$-norm, or max-norm, of the vector of local error measures. A labeling that minimizes Eq.~\ref{eq_pnorm_problem} with $p$ approaching infinity is a \emph{strict minimizer} in the sense of Levi and Zorin~\cite{levi2014strict}. 

It is easily seen that minimizing Eq.~\ref{eq_pnorm_problem} is equivalent to minimizing 

\begin{equation}
\sum_{i\in \mathcal{V}} \phi_i^p(x_i)+\sum_{i,j\in \mathcal{E}} \phi^p_{ij}(x_i,x_j) \; ,
\end{equation}

\noindent i.e., minimizing the sum of all unary and pairwise terms raised to the power $p$. Again, this labeling problem can be solved using minimal graph cuts, provided that all pairwise terms $\phi^p_{ij}$ are submodular. Unfortunately, submodularity of $\phi_{ij}$ does not in general imply submodularity of  $\phi^p_{ij}$.\footnote{As a counterexample, consider the two-label pairwise term $\phi$ given by $\phi(0,0)=3$, $\phi(1,1)=0$, and $\phi(0,1)=\phi (1,0)=2$. It is easily verified that $\phi$ is submodular, while $\phi^2$ is not. }

The contribution of this paper is to provide useful conditions under which $\phi^p_{ij}$ is submodular. Specifically, we show that if $\phi_{ij}$ is submodular and 

\begin{equation}
\max(\phi_{ij}(0,0),\phi_{ij}(1,1)) \leq \max(\phi_{ij}(1,0),\phi_{ij}(0,1)) \; ,
\end{equation}
\noindent then $\phi^p_{ij}$ is submodular for all $p\geq 1$. 

\section{Related Work}
Several authors have considered the use of graph cuts for solving $l_p$-norm optimization problems in image processing, mainly in the context of image segmentation.  In this application, a cut is usually computed directly on a \emph{pixel adjacency graph} -- a graph whose vertex set is the image pixels and where adjacent pixels are connected by weighted edges -- augmented with two vertices ($s$ and $t$) representing object and background labels~\cite{boykov2006graph}. Compared to the objective function given in Eq.~\ref{eq:canonical}, this case only covers a simplified form of the pairwise terms: A fixed penalty is given when two adjacent pixels are assigned different labels, and zero penalty is assigned if the labels are the same. In this simplified case, the issue of submodularity is not important: To optimize the $l_p$ norm of the cut, one may simply raise all \emph{edge weights} to the power $p$ and compute cut as usual. For more general optimization problems however, a pairwise term may assign different penalties to all possible label configurations (for $K$ labels, there are $K^2$ possible label configurations for each pairwise term). This flexibility in assigning the penalties is important in many applications, e.g., stereo reconstruction and image registration.  

All\`{e}ne et al.~\cite{allene2010some} established links relating minimal graph cuts to optimal spanning forests, showing that when the power of the weights of the graph is above a certain number, the cut minimizing the graph cuts energy is a cut by maximum spanning forest. Similar results were independently derived by Miranda et al.~\cite{miranda2009links}. Couprie et al. showed that both methods are instances of an even more general segmentation framework, which they refer to as \emph{power watersheds}~\cite{couprie2011power}. These interesting results all point to the choice of $l_p$-norm being a potentially important hyper-parameter to tune for optimization problems occurring in image analysis and computer vision. The results presented here facilitates the use of minimal graph cuts for solving more general $l_p$-norm problems, beyond the simplified case commonly considered in segmentation applications. 

\section{Conditions for the Submodularity of $\phi^p$}
This section presents our main result; conditions for the submodularity of $\phi^p$. We start by establishing a lemma that is central to the definition of this condition. 

\begin{lemma}
\label{lemma:mainresult}
Let $a,b,c,d,p  \in \mathbb{R}$, with $p> 1$ and $a,b,c,d \geq 0$. If $a+b\leq c+d$ and $\max(a,b) \leq \max(c,d)$ then $a^p+b^p\leq c^p+d^p$.
\end{lemma}
\begin{proof}
Showing that $a^p+b^p\leq c^p+d^p$ is equivalent to showing that $a^p+b^p- c^p-d^p\leq 0$. We assume, without loss of generality, that $a \geq b$ and $c \geq d$ so that $\max(a,b) =a$ and $\max(c,d) =c$.

If $b< d$ then $b^p < d^p$ and, since also $a^p\leq c^p$, the lemma trivially holds. For the remainder of the proof, we will therefore assume that $b\geq d$. It then holds that $c\geq a \geq b \geq d$. 

If $c=0$, then also $a=b=c=0$ and the lemma holds. For the remainder of the proof, we will therefore assume that $c>0$.

If $a+b<c+d$, then $c+d-a-b >0$ and so $a<a+(c+d-a-b)=c+d-b$. Let $A=c+d-b$. Since $d-b\leq0$, it holds that $c \geq A$. Thus the numbers $A, b,c,d$ satisfy the conditions given in the lemma:  $A+b= c+d$ and $\max(A,b) \leq \max(c,d)$. Since $a^p+b^p \leq A^p+b^p$ it follows that if the lemma holds for $A,b,c,d$ then it also holds for $a,b,c,d$. For the remainder of the proof, we will therefore assume that $a+b=c+d$. It follows that $b=c+d-a$ and so 

\begin{equation}
\label{eq:checkpoint}
a^p+b^p- c^p-d^p = a^p+(c+d-a)^p - c^p-d^p \; .
\end{equation}

From the assumption $a\geq b$, it follows that $(c+d)/2\leq a \leq c$. Let 

\begin{equation}
f(x) =x^p +(c+d -x)^p
\end{equation}

\noindent  be a function defined on the domain $x\in [(c+d)/2,c]$. We have 

\begin{equation}
f'(x)=px^{p-1} - p(c+d-x)^{p-1} \; .
\end{equation}

\noindent and 

\begin{equation}
f''(x)=(p-1)px^{p-2} + (p-1)p(c+d-x)^{p-2}  \; .
\end{equation}

\noindent Setting $f'(x)=0$ yields

\begin{align}
&px^{p-1} - p(c+d-x)^{p-1} =0 \\
&\Leftrightarrow px^{p-1} = p(c+d-x)^{p-1} \\
&\Leftrightarrow x = c+d-x \\
&\Leftrightarrow x=(c+d)/2 \; .
\end{align}

The function $f(x)$ thus has a single stationary point at $x=(c+d)/2$ which coincides with the lower bound of the function domain. Since 

\begin{equation}
f''((c+d)/2)=2(p-1)p((c+d)/2)^{p-2}>0
\end{equation}

\noindent this stationary point is a local minimum. Therefore, the maximum of $f(x)$ is attained at the upper bound of the domain $x=c$, and so $f(x)\leq f(c)=c^p+d^p$ on its domain.

Returning to Eq.~\ref{eq:checkpoint}, we now have

\begin{align}
a^p+b^p- c^p-d^p &= a^p+(c+d-a)^p - c^p-d^p \\
&=f(a) -c^p -d^p \\
&\leq c^p + d^p-c^p -d^p\\
&=0 \; . 
\end{align}

\noindent This concludes the proof. 

\end{proof}

\begin{theorem}
Let $\phi$ be a submodular pairwise term. If $\max(\phi(0,0), \phi(1,1))\leq \max(\phi(1,0), \phi(0,1))$, then $\phi^p$ is also submodular, for any real $p\geq 1$. 
\end{theorem}
\begin{proof}
Taking $a=\phi(0,0)$, $b=\phi(1,1)$, $c=\phi(1,0)$ and $d=\phi(0,1)$, the theorem follows directly from Lemma~\ref{lemma:mainresult}.
\end{proof}

\section{Conclusions}
We have presented a condition under which a pairwise term $\phi^p$ is submodular for all $p\geq 1$, thereby identifying a large class of $l_p$-norm objective functions that can be minimized via minimal graph cuts. The conditions derived here are easy to verify for a given set of pairwise terms, and thus make it easier to apply minimal graph cuts for solving  labeling problems with $l_p$-norm objective functions, without having to explicitly prove the submodularity of the pairwise terms for each specific $p$.

It should be noted that even when there are non-submodular pairwise terms, graph cut techniques may still be used to find approximate solutions~\cite{kolmogorov2007minimizing}. Nevertheless, submodularity remains an important property for determining the feasibility of optimizing labeling problems via minimal graph cuts. 

%
%
\bibliographystyle{splncs03}
\bibliography{refs}

\end{document}